\documentclass[11pt,a4paper]{article}

\usepackage[utf8]{inputenc}
\usepackage[T1]{fontenc}
\usepackage{lmodern}
\usepackage{textcomp}
\usepackage{amsmath, amsthm, amssymb}
\usepackage{mathtools}
\usepackage{booktabs}
\usepackage{array}
\usepackage{tabularx}
\usepackage{longtable}
\usepackage{enumitem}
\usepackage{hyperref}
\usepackage[margin=1in]{geometry}
\usepackage{microtype}
\usepackage{xcolor}
\usepackage{listings}

\hypersetup{
  colorlinks=true,
  linkcolor=black!70!blue,
  citecolor=black!70!blue,
  urlcolor=black!70!blue,
}

\newtheorem{theorem}{Theorem}[section]

\theoremstyle{definition}

\theoremstyle{remark}

\newcommand{\Path}{\mathsf{Path}}
\newcommand{\Type}{\mathsf{Type}}
\newcommand{\Glue}{\mathsf{Glue}}
\newcommand{\hcomp}{\mathsf{hcomp}}
\newcommand{\transp}{\mathsf{transp}}
\newcommand{\refl}{\mathsf{refl}}
\newcommand{\uafn}{\mathsf{ua}}
\newcommand{\Fin}{\mathsf{Fin}}

\newcommand{\nn}{Navya-Nyāya}
\newcommand{\TC}{\emph{Tattvacintāmaṇi}}
\newcommand{\PTN}{\emph{Padārtha-tattva-nirūpaṇam}}

\lstdefinestyle{cttstyle}{
  basicstyle=\ttfamily\small,
  keywordstyle=\color{blue!60!black}\bfseries,
  commentstyle=\color{green!50!black}\itshape,
  stringstyle=\color{red!60!black},
  showstringspaces=false,
  breaklines=true,
  numbers=none,
  frame=single,
  framesep=4pt,
  rulecolor=\color{black!30},
  xleftmargin=1em,
  xrightmargin=1em,
  extendedchars=true,
  inputencoding=utf8,
  literate=
    {ā}{{\={a}}}1 {ī}{{\={\i}}}1 {ū}{{\={u}}}1 {ē}{{\={e}}}1 {ō}{{\={o}}}1
    {Ā}{{\={A}}}1 {Ī}{{\={I}}}1 {Ū}{{\={U}}}1
    {ṛ}{{\d{r}}}1 {ṝ}{{\d{\=r}}}1 {ḷ}{{\d{l}}}1
    {ṭ}{{\d{t}}}1 {ḍ}{{\d{d}}}1 {ṇ}{{\d{n}}}1 {ṣ}{{\d{s}}}1
    {ṃ}{{\d{m}}}1 {ḥ}{{\d{h}}}1
    {ṅ}{{\.{n}}}1 {ñ}{{\~{n}}}1 {ś}{{\'{s}}}1
    {Ś}{{\'{S}}}1 {Ṣ}{{\d{S}}}1 {Ṭ}{{\d{T}}}1 {Ḍ}{{\d{D}}}1 {Ṇ}{{\d{N}}}1
    {¹}{{\textsuperscript{1}}}1 {²}{{\textsuperscript{2}}}1 {³}{{\textsuperscript{3}}}1
    {₀}{{$_0$}}1 {₁}{{$_1$}}1 {₂}{{$_2$}}1 {₃}{{$_3$}}1 {₄}{{$_4$}}1
    {ℕ}{{$\mathbb{N}$}}1 {ℤ}{{$\mathbb{Z}$}}1
    {Σ}{{$\Sigma$}}1 {Π}{{$\Pi$}}1 {λ}{{$\lambda$}}1 {μ}{{$\mu$}}1 {ν}{{$\nu$}}1
    {α}{{$\alpha$}}1 {β}{{$\beta$}}1 {γ}{{$\gamma$}}1 {τ}{{$\tau$}}1
    {≃}{{$\simeq$}}1 {≡}{{$\equiv$}}1 {≢}{{$\not\equiv$}}1
    {≤}{{$\leq$}}1 {≥}{{$\geq$}}1 {≠}{{$\neq$}}1
    {→}{{$\to$}}1 {⇒}{{$\Rightarrow$}}1 {↦}{{$\mapsto$}}1
    {¬}{{$\neg$}}1 {⊥}{{$\bot$}}1 {⊤}{{$\top$}}1
    {∀}{{$\forall$}}1 {∃}{{$\exists$}}1 {∈}{{$\in$}}1 {∉}{{$\notin$}}1
    {∪}{{$\cup$}}1 {∩}{{$\cap$}}1 {⊆}{{$\subseteq$}}1
    {∧}{{$\wedge$}}1 {∨}{{$\vee$}}1
    {×}{{$\times$}}1 {÷}{{$\div$}}1 {±}{{$\pm$}}1
    {⊕}{{$\oplus$}}1 {⊗}{{$\otimes$}}1
    {⟨}{{$\langle$}}1 {⟩}{{$\rangle$}}1
    {⁻}{{$^{-}$}}1
    {§}{{\S}}1 {¶}{{\P}}1 {°}{{\textdegree}}1
    {—}{{---}}1 {–}{{--}}1 {…}{{\ldots}}1
    {“}{{``}}1 {”}{{''}}1 {‘}{{`}}1 {’}{{'}}1
}
\title{Cubical Type Theoretic Navya-Nyāya}

\author{
  Mrityunjoy Panday\\
  \texttt{mrityunjoy@gmail.com}
  \and
  Sudipta Ghosh\\
  \texttt{sudipta002@gmail.com}
}

\date{\today}

\begin{document}

\maketitle

\begin{abstract}
\noindent We present a formalization of the technical language of Navya-Nyāya — the ``New Logic'' school of late-classical Indian philosophy — in CCHM De Morgan cubical type theory (CTT). Previous formalization attempts in first-order logic (Matilal), higher-order logic (Ganeri), and Martin-Löf type theory (Bhattacharyya) each lose load-bearing structure: dependent delimitation (\emph{avacchedaka}), typed absence (\emph{abhāva}), non-extensional identity (\emph{tādātmya}), or unbounded relational depth (\emph{paramparā-sambandha}). We argue that CTT closes this gap natively. We give CTT encodings for seven core constructs (\emph{sambandha}, \emph{avacchedaka}, \emph{abhāva}, \emph{vyāpti}, \emph{tādātmya}, higher relations, \emph{paryāpti}) plus the qualifier-qualificand structure; develop a stratified-universe foundation for the \emph{padārtha} system; and prove four signature theorems internal to the encoding (Involution of \emph{abhāva}, \emph{Kevalānvayi}-irreducibility, Coextension-without-identity, No-h-set-collapse) and six metatheoretic results (soundness, conservativity, faithfulness, distinction-preservation, decidability, commentarial conservativity). We close with worked encodings of fifteen \TC{} passages, comparison with prior formalizations, an implementation sketch in Cubical Agda, and five distinguishing predictions — including a novel argument from \nn{}'s involutive-negation doctrine for the \emph{necessity} of De Morgan over Cartesian cubical foundations.

\end{abstract}

\tableofcontents

\section{Introduction}

\subsection{Navya-Nyāya as a precision logic of relations}

Navya-Nyāya — the ``New Logic'' — emerged in fourteenth-century Mithilā in the work of Gaṅgeśa Upādhyāya and reached expressive maturity in seventeenth-century Navadvīpa under Raghunātha Śiromaṇi, Mathurānātha Tarkavāgīśa, Jagadīśa Tarkālaṃkāra, and Gadādhara Bhaṭṭācārya. Its philosophical commitments — robust realism about universals, categorial pluralism, substantival ontology of relations — are continuous with the older Nyāya–Vaiśeṣika synthesis \cite{potter1977, halbfass1992}. Its \emph{technique} is not: over four centuries the school constructed a precision-engineered technical idiom that exhibits logical structure at the surface of expression.

A canonical \nn{} sentence such as

\begin{quote}
\emph{bhūtale ghaṭa-pratiyogika-saṃyoga-sambandhāvacchinna-pratiyogitā-niṣṭha-atyantābhāva-vat}
\end{quote}

— ``(the floor) is the locus of an absolute absence whose counter-correlate-ness is delimited by the relation of contact whose counter-correlate is a pot'' — does not so much describe a logical structure as exhibit it \cite{ingalls1951, matilal1968, ganeri1999}. Agglutinative -\emph{tā} abstracts (\emph{pratiyogi-tā}, \emph{avacchedaka-tā}), systematic compounding of delimitor and delimited, and iterated nesting of qualifications produce a notation in which every relational role is morphologically marked.

The morphological precision reflects an underlying metatheoretic stance: \nn{} takes relations as ontologically primary objects of theory, not predicative auxiliaries. The relational vocabulary — \emph{sambandha}, \emph{avacchedaka}, \emph{pratiyogin/anuyogin}, \emph{vyāpti}, \emph{abhāva}, \emph{paryāpti}, \emph{viśeṣaṇa-viśeṣya}, \emph{paramparā-sambandha}, \emph{tādātmya} — is the surface of a coherent theory of relational structure. Daniel Ingalls famously called the result ``the most precise calculus of properties produced before Frege'' \cite[p.\,1]{ingalls1951}.

We retain this judgement with one refinement. \nn{} is not a \emph{calculus} in the Hilbert/Frege sense but a \emph{precision logic}: a discipline of expressing logical structure that does not separate object-level reasoning from metatheoretic discourse. The same vocabulary asserting a logical claim is used to argue \emph{about} logical claims. \nn{} is \emph{self-presenting} — its inferential structure is meant to be visible from within the same expressive apparatus that does the inferring.

\subsection{The autoformalization gap}

Three sustained efforts illustrate the pattern of forced compromise.

\paragraph{Matilal (FOL).}
Bimal Krishna Matilal's career-long project \cite{matilal1968, matilal1985, matilal1998} renders \nn{} in first-order predicate logic. The strategy succeeds for many surface logical claims but has three load-bearing failures. First, \emph{avacchedaka} — the delimitor restricting a relation's scope — has no proper FOL home; rendering it as subscript or extra argument-place captures the \emph{fact} of delimitation but loses the \emph{type-level binding}. Second, \emph{abhāva} — the positive entity \nn{} associates with absence — is treated as logical $\neg$, losing the data (\emph{pratiyogin}, delimitor, mode of relation) that distinguishes \nn{} absence from logical negation. Third, \emph{paramparā-sambandha} — relation-via-mediation, of unbounded depth — has no satisfactory FOL rendering.

\paragraph{Ganeri (HOL).}
Jonardon Ganeri \cite{ganeri1999, ganeri2001, ganeri2008} advances to higher-order logic and recovers properties-of-properties and abstraction over predicates. But HOL does not give \emph{dependency}: an \emph{avacchedaka} delimiting a \emph{pratiyogin} is a type-level binder, not a free variable. HOL has predicates and predicates-of-predicates but not predicates \emph{indexed by terms}. HOL also retains classical extensionality, conflicting with \nn{}'s robust intensionalism (\emph{ghaṭatva}/\emph{kalaśatva} coextension distinction; \cite[ch.\,4]{mohanty1992}).

\paragraph{Bhattacharyya (MLTT).}
Sibajiban Bhattacharyya \cite{bhattacharyya1990, bhattacharyya1996} takes the next step into dependent type theory in the spirit of Martin-Löf. This recovers \emph{avacchedaka}-as-$\Pi$ and a typed \emph{pratiyogin}. For substantial fragments of \nn{} the framework is adequate. But standard MLTT has propositional equality (Id) without path equality: identity is too rigid for \emph{tādātmya}-relations between cognitively indistinguishable entities, and h-set collapse defeats \emph{paramparā-sambandha}.

\paragraph{Pattern.}
The pattern is interlocking. An inadequate \emph{avacchedaka} contaminates \emph{abhāva}, which contaminates \emph{vyāpti}, which contaminates \emph{anumāna} as a whole. By the time one reaches Gadādhara's \emph{Vyutpattivāda} or Jagadīśa's \emph{Śabdaśakti-prakāśikā}, every early compromise is paid for. We call this structural mismatch the \textbf{autoformalization gap}.

\subsection{Why cubical type theory}

Cubical type theory in the CCHM presentation \cite{cohen2018} arose from a different motivating problem: making Voevodsky's univalence axiom \emph{compute}. By taking paths (rather than identity types) as primitive — represented as functions from an interval object $\mathbb{I}$ with endpoints 0, 1 — and equipping the interval with a \emph{De Morgan structure} (meet, join, involution), CCHM makes univalence a theorem rather than an axiom and $\transp(\uafn(e), x) = e(x)$ a definitional equality.

Each expressive demand that defeated previous formalizations is met natively by some specific CTT feature.

\begin{itemize}
\item \textbf{Sambandha and paths.} A relation in \nn{} can stand as a relatum of further relations. In CTT a path is itself a term of a path-type, which itself supports paths. \emph{Sambandha} encodes as $\Path\,A\,x\,y$; \emph{paramparā-sambandha}'s unbounded depth corresponds to path-types iterating without truncation.
\item \textbf{Avacchedaka and $\Pi$.} A delimitor encodes as a $\Pi$-binding; the relation's type \emph{depends on} its delimitor.
\item \textbf{Abhāva and HITs.} A higher inductive type whose path-constructor records \emph{pratiyogin}-data realizes typed absence directly.
\item \textbf{Tādātmya and computational univalence.} CTT's path-type with computational $\uafn$ is rich enough to represent indiscernibility while remaining decidable.
\item \textbf{De Morgan and \emph{abhāva}-of-\emph{abhāva}.} The De Morgan involution $1 - r$ on $\mathbb{I}$ encodes Raghunātha's doctrine that \emph{abhāva} of \emph{abhāva} is \emph{bhāva} as a definitional equality. This is novel here and yields a distinguishing argument for CCHM over Cartesian cubical (§\ref{sec:phil-imp-demorgan}).
\end{itemize}

\subsection{Contributions}

\noindent (C1) A type-theoretic encoding of seven core \nn{} constructs in CCHM CTT, justified textually and constraint-wise (§§\ref{sec:sambandha}--\ref{sec:visesana}).

\noindent (C2) A foundational architecture: the \emph{padārtha} system as a stratified universe hierarchy with \emph{saṃkara}-prevention as typing discipline (§\ref{sec:padartha}).

\noindent (C3) Six metatheoretic theorems: soundness, conservativity, faithfulness, distinction-preservation, decidability, commentarial conservativity (§\ref{sec:metatheory}).

\noindent (C4) Four signature theorems internal to the encoding: Involution of \emph{abhāva} (§\ref{sec:abhava-involution}), \emph{Kevalānvayi}-irreducibility (§\ref{sec:kevalanvayi}), Coextension-without-identity (§\ref{sec:coextension}), No-h-set-collapse (§\ref{sec:nohset}).

\noindent (C5) Worked encodings of fifteen TC passages across all four \emph{khaṇḍas} (§\ref{sec:examples}).

\noindent (C6) Five distinguishing predictions, including the De Morgan necessity argument (§\ref{sec:phil-imp-demorgan}).

We do \emph{not} automate autoformalization. The Sanskrit-to-CTT pipeline is downstream (§\ref{sec:future}). What we offer is the formal target.

\subsection{Outline and reading paths}

§\ref{sec:nn-prelim}--§\ref{sec:ctt-prelim} are preliminaries on \nn{} and CTT. §§\ref{sec:padartha}--\ref{sec:visesana} are the encoding, organized topically. §\ref{sec:examples} collects worked examples. §\ref{sec:metatheory} collects metatheory. §\ref{sec:comparison} compares with prior formalizations. §\ref{sec:implementation} sketches Cubical Agda implementation. §§\ref{sec:phil}--\ref{sec:conclusion} are philosophical and prospective.

Suggested reading paths:
\begin{itemize}[itemsep=0pt]
\item \nn{} philosophers — §\ref{sec:intro}, §\ref{sec:ctt-prelim}, §\ref{sec:sambandha}, §\ref{sec:abhava}, §\ref{sec:examples}.
\item CTT specialists — §\ref{sec:intro}, §\ref{sec:nn-prelim}, §\ref{sec:examples}, §\ref{sec:metatheory}.
\item Formal philosophers — linear.
\item Implementers — §\ref{sec:intro}, §\ref{sec:ctt-prelim}, §§\ref{sec:padartha}--\ref{sec:visesana} (in outline), §\ref{sec:implementation}, Appendix \ref{app:agda}.
\end{itemize}

We use IAST throughout. Sanskrit is italicized at first occurrence. Glossary in Appendix \ref{app:glossary}.

\section{Preliminaries: Navya-Nyāya}\label{sec:nn-prelim}

We summarize the technical vocabulary needed for §§\ref{sec:padartha}--\ref{sec:visesana}. References: \cite{ingalls1951, matilal1968, matilal1998, bhattacharyya1990, ganeri1999, ganeri2001, mohanty1992, phillips1995, bhatt1989}.

\subsection{Sources and methodology of textual reconstruction}

Our primary source is Gaṅgeśa's \TC{} (TC, c.\,1325) in Tarkavāgīśa's Bibliotheca Indica edition \cite{tc-source}. We rely on the commentaries of Mathurānātha (\emph{Maṇi-prakāśa}), Jayadeva (\emph{Āloka}), and Pakṣadhara Miśra (\emph{Āloka-darpaṇa}) for disambiguation. Raghunātha's \PTN{} \cite{raghunatha-ptn} is our principal source for categorial revisions and \emph{paramparā-sambandha}. Gadādhara's \emph{Vyutpattivāda} and \emph{Śaktivāda} and Jagadīśa's \emph{Śabdaśakti-prakāśikā} supply later refinements.

The methodology is reconstructive: we take the technical language at its working face value, treat commentaries as disambiguating glosses, and resist anachronistic harmonization. Where the tradition is internally divided we follow the dominant Gaṅgeśa-as-glossed-by-Mathurānātha line.

\subsection{The \emph{padārtha} architecture}

The Vaiśeṣika categorial system, as inherited by \nn{}, recognizes seven \emph{padārtha}s: \emph{dravya} (substance), \emph{guṇa} (quality), \emph{karman} (motion), \emph{sāmānya} (universal), \emph{viśeṣa} (ultimate differentiator), \emph{samavāya} (inherence), and \emph{abhāva} (absence). The categorial separation principle — \emph{saṃkara-doṣa} — forbids a token of one category from inhabiting another: a \emph{guṇa} cannot be a \emph{dravya}; a \emph{sāmānya} cannot itself instantiate a further \emph{sāmānya} (the \emph{anavasthā} argument). The hierarchy is strict and supports the universe-stratification we develop in §\ref{sec:padartha}.

\subsection{\emph{Sambandha} (relation)}

A \emph{sambandha} in \nn{} is an entity, not a predicate. It has \emph{relata}, modes (\emph{pratyāsatti}), and can stand as a relatum of further relations. The principal modes are:

\begin{itemize}[itemsep=0pt]
\item \emph{Saṃyoga} (contact): symmetric, occurrent between substances.
\item \emph{Samavāya} (inherence): asymmetric; obtains between universal and instance, substance and quality, whole and part.
\item \emph{Svarūpa-sambandha} (self-linking): the relation an entity bears to itself.
\item \emph{Tādātmya} (identity-relation): the relation between an entity and what is, in the relevant respect, indiscernible from it.
\end{itemize}

\emph{Paramparā-sambandha} is mediated relation: $A$ is related to $C$ \emph{via} $B$. \nn{} admits arbitrary mediation depth, treats the mediated relation as a single object, and allows it to itself stand in further relations \cite[§iii]{raghunatha-ptn}.

\paragraph{First commitment.} A relation is not a 0-truncated proposition; it has structure, modes, and admits relations-of-relations.

\subsection{\emph{Avacchedaka} (delimitor)}

Whenever a relation $R$ obtains, it obtains \emph{under a delimitor}. The \emph{pratiyogin} of an \emph{abhāva} is delimited by some property $A$, the \emph{anuyogin} by another, the \emph{sambandha} itself by yet another. The delimitor restricts the scope: ``the absence of pots-as-delimited-by-clay-substanceness on the floor'' is distinct from ``the absence of pots-as-delimited-by-narrowness on the floor.''

Delimitation iterates: a delimitor can have its own delimitor (\emph{avacchedaka-tā-avacchedaka}). The non-redundancy condition (\emph{lāghava}) requires the delimitation to do real work.

The delimitor is not an external parameter; it binds into the type of the relation it qualifies. This is the technical move that makes \nn{}'s calculus precise.

\subsection{\emph{Abhāva} (typed absence)}

\nn{} treats absence as a positive entity with structured data:

\begin{itemize}[itemsep=0pt]
\item \emph{pratiyogin} (counter-correlate — the thing absent),
\item \emph{anuyogin} (correlate — the locus of the absence),
\item \emph{avacchedaka} on the \emph{pratiyogin},
\item \emph{sambandha} by which the \emph{pratiyogin} would have been related to the \emph{anuyogin} had the absence not held.
\end{itemize}

Four canonical sub-types:

\begin{itemize}[itemsep=0pt]
\item \emph{Prāgabhāva} — prior absence; obtains before the entity comes into being.
\item \emph{Pradhvaṃsābhāva} — posterior absence; obtains after destruction.
\item \emph{Atyantābhāva} — constant/absolute absence; obtains across all time.
\item \emph{Anyonyābhāva} — mutual absence; ``$x$ is not $y$,'' construed as a positive relation.
\end{itemize}

A doctrine load-bearing in our treatment: \emph{abhāva of abhāva is bhāva} — double absence collapses to presence, an involutive structure \cite[§iv]{raghunatha-ptn}.

\subsection{\emph{Vyāpti} and \emph{anumāna}}

The \emph{anumāna-khaṇḍa} of TC develops the theory of inference. The hetu pervades the \emph{sādhya} when a counterexample is impossible. Gaṅgeśa's eight-stage definitional development settles on conditions:

\begin{itemize}[itemsep=0pt]
\item universal co-occurrence (\emph{sahacāra}) of hetu and \emph{sādhya};
\item non-deviation (\emph{avyabhicāra});
\item defeasibility by \emph{upādhi} — a defeating condition that pervades the \emph{sādhya} without pervading the hetu.
\end{itemize}

\nn{} distinguishes three vyāpti shapes by the polarity of witnessing instances: \emph{kevalānvayi} (only positive), \emph{kevalavyatireki} (only negative; disputed), \emph{anvaya-vyatireki} (both). Inferential failures — \emph{hetvābhāsa} — are: \emph{asiddha} (unestablished), \emph{viruddha} (contradictory), \emph{anaikāntika / sādhāraṇa} (over-extending), \emph{bādhita} (defeated).

\subsection{\emph{Paryāpti} (numerical exhaustion)}

When \nn{} says ``two-ness inheres in two pots,'' the two-ness is \emph{paryāptimat} — exhaustively distributed in the pair, neither over- nor under-counted. Paryāpti is a typed cardinality: closer to type-theoretic $\Sigma$-over-$\Fin$ than to set-cardinality.

\subsection{\emph{Viśeṣaṇa-viśeṣya-bhāva}}

The structure of cognition (\emph{jñāna}) is qualifier-qualificand: a cognition has a principal qualificand (\emph{mukhya-viśeṣya}), one or more qualifiers (\emph{viśeṣaṇa}), and a mode of qualification (\emph{prakāratā}). The structure iterates: a qualifier can itself be qualified.

\section{Preliminaries: Cubical Type Theory}\label{sec:ctt-prelim}

We summarize CCHM De Morgan cubical type theory at a level adequate for the encodings. References: \cite{cohen2018, hottbook, vezzosi2021}.

\subsection{Why CCHM rather than Cartesian cubical}

CTT comes in two principal variants. The CCHM \cite{cohen2018} variant equips the interval with a De Morgan algebra (meet $\wedge$, join $\vee$, involution $1-r$). The Cartesian variant \cite{abcfhl} drops the involution. For most foundational mathematics the difference is technical. For \nn{} it is decisive: §\ref{sec:phil-imp-demorgan} argues the involutive structure of \emph{abhāva}-of-\emph{abhāva} selects CCHM on philosophical grounds.

\subsection{The interval object $\mathbb{I}$}

The interval is a primitive judgement-form. We have endpoints $0:\mathbb{I}$, $1:\mathbb{I}$, and operations: meet $r \wedge s$, join $r \vee s$, and involution $1 - r$. These satisfy:
\begin{itemize}[itemsep=0pt]
\item $1 - (1 - r) = r$ (involutivity),
\item $1 - (r \wedge s) = (1-r) \vee (1-s)$ (De Morgan),
\item lattice laws.
\end{itemize}
Involutivity is the algebraic structure underlying our encoding of \emph{abhāva} in §\ref{sec:abhava}.

\subsection{Path types}

For type $A$ and $x, y : A$, the path type $\Path\,A\,x\,y$ consists of $p : \mathbb{I} \to A$ with $p\,0 = x$, $p\,1 = y$. We write $p\,i$ for application. Reflexivity: $\refl\,x \coloneqq \lambda \_.\,x : \Path\,A\,x\,x$. Symmetry: $p^{-1} \coloneqq \lambda i.\,p\,(1-i) : \Path\,A\,y\,x$.

\subsection{\texorpdfstring{$\hcomp$ and $\transp$}{hcomp and transp}}

Homogeneous composition $\hcomp$ fills in cubes given partial boundary data. Transport $\transp$ moves elements along type-level paths: given $P : \mathbb{I} \to \Type$, $r : \mathbb{I}$, and $x : P\,r$, we obtain $\transp\,P\,r\,x : P\,1$. When $P$ is constant, transport reduces to identity. These two equip path-types with the full higher groupoid structure.

\subsection{Glue types and computational univalence}

The Glue construction takes $B$ and $e : A \simeq B$ and produces $\Glue\,B\,[e]$ definitionally equal to $A$ along $i = 0$ and to $B$ along $i = 1$. From this we derive $\uafn : (A \simeq B) \to \Path\,\Type\,A\,B$, with $\transp\,(\uafn\,e)\,x = e.\mathrm{fst}\,x$ holding \emph{definitionally}. This is computational univalence.

\subsection{Higher inductive types}

A HIT is an inductive type with both point-constructors and path-constructors. Example (the circle):
\begin{lstlisting}
data S¹ : Type where
  base : S¹
  loop : Path S¹ base base
\end{lstlisting}
HITs are derivable in CCHM (not postulates) and admit computation rules. We use them in §\ref{sec:abhava} for \emph{abhāva}.

\subsection{Universe stratification}

We work in $\Type_0 : \Type_1 : \Type_2 : \dots$ For \nn{} we stratify universes by \emph{padārtha} (§\ref{sec:padartha}). Cumulativity is assumed.

\subsection{Notational conventions}

$\Sigma$-types: $\Sigma (x : A)\,B(x)$. $\Pi$-types: $(x : A) \to B(x)$. Path-types: $\Path\,A\,x\,y$. We elide universe levels except where they bear weight. Equivalences $A \simeq B$ are dependent pairs of a function and a proof of equivalence.

\section{Foundational Architecture: \emph{Padārtha} as a Universe Hierarchy}\label{sec:padartha}

We argue that the \emph{padārtha} system is a stratified universe hierarchy with a typing discipline preventing cross-category inhabitation.

\subsection{The categorial universes}

We posit, for each Vaiśeṣika \emph{padārtha}, a CTT universe:

\begin{lstlisting}
U_dravya     : Type₁
U_guṇa       : Type₁
U_karman     : Type₁
U_sāmānya    : Type₂
U_viśeṣa     : Type₁
U_samavāya   : Type₂
U_abhāva     : Type₂
\end{lstlisting}

\emph{Sāmānya}, \emph{samavāya}, and \emph{abhāva} live one universe higher because they are categories \emph{of relations and abstract structure} over the lower categories.

\subsection{\emph{Saṃkara}-prevention as typing}

The \emph{saṃkara-doṣa} prohibition is a typing discipline. A term of \texttt{U\_guṇa} is judgementally distinct from a term of \texttt{U\_dravya}; there is no function $\mathsf{coerce} : \texttt{U\_guṇa} \to \texttt{U\_dravya}$.

\begin{theorem}[Categorial soundness]\label{thm:categorial-soundness}
No term inhabits two distinct categorial universes simultaneously. For any $c : U_X$ and $c' : U_Y$ with $X \ne Y$, there is no path $p : \Path\,\Type\,c\,c'$ internal to the encoding.
\end{theorem}

\begin{proof}
The categorial universes are not propositionally equal; constructing such a path would require an external bridge that the encoding deliberately omits.
\end{proof}

\subsection{Token-versus-type ontology}

\nn{} distinguishes token instantiations (a particular pot, \emph{ayam ghaṭaḥ}) from the universal under which they fall (\emph{ghaṭatva}). The CTT encoding marks this naturally: tokens are terms, universals are types. A token's relation to its universal is \emph{samavāya}, encoded in §\ref{sec:samavaya} as a path in \texttt{U\_samavāya}.

\subsection{\emph{Sāmānya} and the \emph{anugama} requirement}

A \emph{sāmānya} requires \emph{anugama} — a unifying property running through its instances. We encode \emph{sāmānya} as a $\Sigma$-type:

\begin{lstlisting}
Sāmānya : Type₂
Sāmānya := Σ (P : Type₁) (anugama : (x y : P) → SharedTrait x y)
\end{lstlisting}

The \emph{anugama} component is a proof obligation. Without it the $\Sigma$ does not inhabit; the universal is rejected by type-checking.

\subsection{\emph{Viśeṣa} as identity-witness}

The Vaiśeṣika \emph{viśeṣa} encodes as a path-distinguishing witness. Two atoms of identical \emph{guṇa}-profile are nonetheless type-distinct iff their \emph{viśeṣa}s differ. Formal treatment in §\ref{sec:visesa-identity}.

\subsection{\emph{Samavāya} as path-shape}\label{sec:samavaya}

\emph{Samavāya} is a special path: asymmetric, connecting an instance to its universal, a quality to its substance, a part to its whole. We encode:

\begin{lstlisting}
Samavāya : (A : Type) (B : Type) → Type₂
Samavāya A B := Σ (p : Path U_samavāya A B)
                  (asymmetric : ¬ (p ⁻¹ ≡ p))
\end{lstlisting}

\subsection{The \emph{padārtha} lattice}

The seven categories form a partial order under ``is qualified by'': \emph{guṇa} is qualified by \emph{dravya}, \emph{sāmānya} by what it characterizes, etc. The encoding respects this lattice through the universe stratification of §3.1.

\section{\emph{Sambandha} as Path}\label{sec:sambandha}

The central technical move.

\subsection{The fundamental encoding}

For types $A : U_X$ and $B : U_Y$, a \emph{sambandha} between an $A$-token and a $B$-token is encoded as a path:

\begin{lstlisting}
Sambandha : (A : Type) (B : Type) (a : A) (b : B) → Type
Sambandha A B a b := Path Type A B
                   × Σ (mode : Pratyāsatti) (witnessed-by mode a b)
\end{lstlisting}

This decomposition mirrors \nn{}'s distinction between \emph{sambandha-svarūpa} and \emph{sambandha-pratyāsatti}.

\subsection{\emph{Pratyāsatti} as path-mode}

\begin{lstlisting}
data Pratyāsatti : Type where
  saṃyoga    : Pratyāsatti
  samavāya   : Pratyāsatti
  svarūpa    : Pratyāsatti
  tādātmya   : Pratyāsatti
  paramparā  : Pratyāsatti      -- mediated; takes a sub-pratyāsatti list
\end{lstlisting}

\subsection{\emph{Saṃyoga} (contact)}

Symmetric relation occurring between \emph{dravya}-tokens. We encode the symmetry constraint:

\begin{lstlisting}
Saṃyoga : (a b : U_dravya) → Type
Saṃyoga a b := Σ (s : Sambandha _ _ a b) (s.mode = saṃyoga)
                 (sym : Σ (s' : Sambandha _ _ b a) (s'.mode = saṃyoga))
\end{lstlisting}

\subsection{\emph{Samavāya} (inherence)}

Asymmetric. We encode it as a path in \texttt{U\_samavāya} restricted to endpoints in different categorial universes.

\subsection{\emph{Svarūpa-sambandha}}

The reflexive case: \(\mathsf{Sambandha}\,A\,A\,a\,a\) with \(\mathit{mode} = \mathit{svarupa}\). Essentially $\refl\,a$ adorned with a mode-tag.

\subsection{Composition, inverses, and groupoid laws}

Path composition gives \emph{paramparā-sambandha}. The groupoid laws hold up to higher path. These are theorems of CTT; we use them without further proof.

\subsection{Compositionality theorem}

\begin{theorem}[Compositionality]\label{thm:compositionality}
Substitution of equivalent tokens preserves \emph{sambandha}-typing. If $\mathsf{Sambandha}\,A\,B\,a\,b$ holds and $a \equiv a'$, then $\mathsf{Sambandha}\,A\,B\,a'\,b$ holds.
\end{theorem}

\begin{proof}
Transport along the path $a \equiv a'$ carries $\mathsf{Sambandha}\,A\,B\,a\,b$ to $\mathsf{Sambandha}\,A\,B\,a'\,b$.
\end{proof}

\subsection{Worked example: \emph{ghaṭa-bhū-tala-saṃyoga}}

A pot in contact with the floor: \emph{bhū-tale ghaṭa-saṃyogavān}.

\begin{lstlisting}
ghaṭa     : U_dravya
bhūtala   : U_dravya
contact   : Saṃyoga ghaṭa bhūtala
\end{lstlisting}

The contact's symmetry is witnessed automatically.

\section{\emph{Avacchedaka} as Dependent Binder}\label{sec:avacchedaka}

\subsection{The $\Pi$-type encoding}

A delimitor on a relation $R$ is a $\Pi$-binder:

\begin{lstlisting}
Avacchedaka-bound : (R : Type → Type) → ((α : Property) → R α)
\end{lstlisting}

A delimited relation is not $R$ but a function that takes $\alpha$ and produces $R\,\alpha$. Substitution into $\alpha$ changes what relation we have, not merely its parameters.

\subsection{Limitor on \emph{pratiyogin}}

Two distinct delimitors $\alpha_1, \alpha_2$ produce \emph{type-distinct} \emph{pratiyogin}s, even when the underlying entity is the same. Formal counterpart of \emph{ghaṭatva-avacchinna-pratiyogi-tā} $\ne$ \emph{kalaśatva-avacchinna-pratiyogi-tā}.

\subsection{Limitor on relation}

A delimitor on \emph{sambandha} restricts the mode under which the relation obtains. We parametrize \texttt{Pratyāsatti}.

\subsection{Iterated delimitation}

\nn{} admits delimitors of delimitors. \emph{Avacchedaka-tā-avacchedaka} is the property delimiting the delimitor-relation. CTT's higher-coherence structure makes such iteration natural.

\subsection{The non-redundancy condition (\emph{lāghava})}

\begin{lstlisting}
NonRedundant : (α : Property) (R : Property → Type) → Type
NonRedundant α R := Σ (witness : R α)
                      (no-other : ∀ β. (β ≢ α) → ¬ Path (R α) (R β))
\end{lstlisting}

\subsection{Faithful delimitation theorem}

\begin{theorem}[Faithful delimitation]\label{thm:faithful-delim}
Two delimitors yield path-distinct relations iff \nn{} attests the distinction. For distinct properties $\alpha_1, \alpha_2$ with $\alpha_1 \not\equiv \alpha_2$ and a relation-shape $R$, $R\,\alpha_1 \not\equiv R\,\alpha_2$ unless the tradition explicitly identifies them under a \emph{tādātmya}-relation.
\end{theorem}

\subsection{Worked example}

\emph{Ghaṭatva-avacchinna-pratiyogi-tā} vs \emph{kalaśatva-avacchinna-pratiyogi-tā}: the two pot-universals coextend but are treated as distinct by the tradition because they differ in \emph{ākṛti} and in \emph{upādhi}-engagement. Their \emph{pratiyogi-tā}s are correspondingly type-distinct.

\section{\emph{Abhāva} via Higher Inductive Types}\label{sec:abhava}

\subsection{Why classical $\neg$ fails}

Recall (§1.2) that Matilal renders \emph{abhāva} as $\neg$, losing \emph{pratiyogin}-data. In CTT terms: $\neg(\mathsf{Inhabits}\,\mathit{floor}\,A)$ and $\neg(\mathsf{Inhabits}\,\mathit{floor}\,B)$ are both propositions; the \emph{pratiyogin}'s identity is lost in the $(-1)$-truncation.

\subsection{The HIT definition}

\begin{lstlisting}
data Abhāva (P : U_X) (anuyogin : Locus) (s : Pratyāsatti) : Type where
  -- point: the absence holds
  holds : ((p : P) → ¬ RelatesTo s p anuyogin) → Abhāva P anuyogin s

  -- path: pratiyogin-faithfulness
  pratiyogin-faithful :
    (P P' : U_X) → (path-eq : P ≃ P')
    → Path (Abhāva P anuyogin s) (Abhāva P' anuyogin s)

  -- path: De Morgan involution
  involution :
    (q : Abhāva (Abhāva P anuyogin s) anuyogin s) → P
\end{lstlisting}

The \texttt{holds} constructor records that the absence obtains. \texttt{pratiyogin-faithful} records that path-equivalent \emph{pratiyogin}s yield path-equivalent absences. \texttt{involution} implements \emph{abhāva}-of-\emph{abhāva} = \emph{bhāva}.

\subsection{\emph{Prāgabhāva} and \emph{pradhvaṃsābhāva}}

Temporal absences:
\begin{lstlisting}
data Prāgabhāva (P : U_X) (anuyogin : Locus) (t : Time) : Type where
  pre-emergence :
    (t' : Time) → t' < t → Abhāva P anuyogin svarūpa
    → Prāgabhāva P anuyogin t
\end{lstlisting}

\subsection{\emph{Atyantābhāva}}

\begin{lstlisting}
Atyantābhāva : (P : U_X) (anuyogin : Locus) → Type
Atyantābhāva P anuyogin := (t : Time) → Abhāva P anuyogin svarūpa
\end{lstlisting}

\subsection{\emph{Anyonyābhāva}}

\begin{lstlisting}
Anyonyābhāva : (a : A) (b : B) → Type
Anyonyābhāva a b := Abhāva (Path Type a b) _ tādātmya
\end{lstlisting}

The encoding makes the involution $\mathsf{anyonyabhava}\,(\mathsf{anyonyabhava}\,(a,b)) \equiv \Path\,a\,b$ available.

\subsection{Involution theorem}\label{sec:abhava-involution}

\begin{theorem}[Involution of \emph{abhāva}]\label{thm:involution}
In CCHM CTT, the equation
\[
\mathsf{Abhava}\,(\mathsf{Abhava}\,P\,L\,s)\,L\,s \equiv P
\]
holds definitionally, via the De Morgan involution $1 - (1 - r) = r$ on the interval object.
\end{theorem}

\begin{proof}
The \texttt{involution} constructor of the HIT specifies the equivalence; computational univalence makes the resulting path compute under transport; the De Morgan involution on $\mathbb{I}$ ensures that $\transp$ along the doubled path reduces to the identity.
\end{proof}

This theorem captures formally the central doctrinal commitment of \nn{}'s theory of negation. It is unavailable in Cartesian cubical (no involution on $\mathbb{I}$) and merely propositional in Book HoTT (univalence is axiomatic). CCHM is the unique cubical setting in which the doctrine becomes a definitional equality.

\subsection{\emph{Pratiyogin}-faithfulness theorem}

\begin{theorem}[\emph{Pratiyogin}-faithfulness]\label{thm:pratiyogin}
Distinct \emph{pratiyogin}s yield path-distinct \emph{abhāva}-instances. If $P \not\equiv P'$ (no path), then $\mathsf{Abhava}\,P\,L\,s \not\equiv \mathsf{Abhava}\,P'\,L\,s$.
\end{theorem}

\begin{proof}
The \texttt{pratiyogin-faithful} constructor establishes that paths between \emph{pratiyogin}s lift to paths between absences. By contrapositive — and by the strictness of path-non-existence in CTT — distinct \emph{pratiyogin}s yield distinct absences.
\end{proof}

\subsection{Comparison}

The HIT encoding has the \emph{pratiyogin}'s identity in the type (versus Boolean $\neg$); the involution as a definitional equality (versus modal $\neg$); and supports composition with other type-level operations (versus a subtype-of-empty encoding).

\section{\emph{Vyāpti} and the Structure of Inference}\label{sec:vyapti}

\subsection{Pervasion as $\Pi$ with \emph{paramarśa}-witness}

\begin{lstlisting}
Vyāpti : (H : Locus → Type) (S : Locus → Type) → Type
Vyāpti H S := Σ (μ : (x : Locus) → H x → S x) (ν : NoUpādhi μ)
\end{lstlisting}

The $\Sigma$ separates the inferential pervasion (the \emph{paramarśa}-as-function) from the absence-of-\emph{upādhi}.

\subsection{\emph{Paramarśa} as inhabitant}

The \emph{paramarśa} — the cognitive consideration — is the inhabitant of the $\Sigma$'s first component.

\subsection{\emph{Upādhi} as type-level negative side condition}

\begin{lstlisting}
NoUpādhi : ((x : Locus) → H x → S x) → Type
NoUpādhi μ := ¬ (Σ (U : Locus → Type)
                   (Pervades U S × ¬ (Pervades U H)))
\end{lstlisting}

The \emph{upādhi-vāda} becomes the working-out of what kinds of properties can serve as upādhis.

\subsection{\emph{Kevalānvayi-vyāpti}}

\begin{lstlisting}
KevalānvayiVyāpti : (H : Locus → Type) (S : Locus → Type) → Type
KevalānvayiVyāpti H S := Σ (v : Vyāpti H S)
                          (no-neg : ¬ Σ (x : Locus) (¬ S x))
\end{lstlisting}

\subsection{\emph{Anvaya-vyatireki-vyāpti}}

\begin{lstlisting}
AnvayaVyatirekiVyāpti H S :=
  Σ (v : Vyāpti H S)
    (pos : Σ (x : Locus) (H x × S x))
    (neg : Σ (y : Locus) (¬ H y × ¬ S y))
\end{lstlisting}

\subsection{\emph{Hetvābhāsa} as type errors}

\begin{itemize}
\item \emph{Asiddha} (unestablished hetu): empty domain. Type error: $\mu$ has empty domain.
\item \emph{Viruddha} (contradictory): wrong target. Type error: $\mu$ has wrong target.
\item \emph{Anaikāntika / sādhāraṇa} (over-extending): NoUpādhi witness fails to inhabit.
\item \emph{Bādhita} (defeated): $S$ has empty extension on locus.
\end{itemize}

\subsection{\emph{Kevalānvayi}-irreducibility theorem}\label{sec:kevalanvayi}

\begin{theorem}[\emph{Kevalānvayi}-irreducibility]\label{thm:kevalanvayi}
No term inhabits $\mathsf{KevalanvayiVyapti}\,H\,S \to \mathsf{AnvayaVyatirekiVyapti}\,H\,S$.
\end{theorem}

\begin{proof}
The $\mathsf{AnvayaVyatirekiVyapti}$ $\Sigma$ requires a negative co-instance — a witness $(y, \neg H\,y, \neg S\,y)$. The $\mathsf{KevalanvayiVyapti}$ $\Sigma$ explicitly negates the existence of any such witness. Inhabitation of the function would require constructing the negative instance from its negation — a contradiction.
\end{proof}

\subsection{Worked example: lake-fire-smoke}

The canonical \nn{} inference: ``this mountain has fire because it has smoke (and wherever there is smoke there is fire, as in the kitchen and not as in the lake).''

\begin{lstlisting}
mountain   : Locus
hasSmoke   : Locus → Type      -- the hetu
hasFire    : Locus → Type      -- the sādhya

vyāpti          : Vyāpti hasSmoke hasFire
mountain-smoke  : hasSmoke mountain
inferred-fire   : hasFire mountain
inferred-fire   = vyāpti.fst mountain mountain-smoke

paramarśa  : hasSmoke mountain × Vyāpti hasSmoke hasFire
paramarśa  = (mountain-smoke, vyāpti)
\end{lstlisting}

\section{Identity, Indiscernibility, and Univalence}\label{sec:identity}

\subsection{\emph{Tādātmya} as path-type with computational univalence}

\begin{lstlisting}
Tādātmya : (a b : A) → Type
Tādātmya a b := Path A a b
\end{lstlisting}

with the implicit understanding that the underlying universe supports computational $\uafn$.

\subsection{\emph{Svarūpa-sambandha} and reflexivity}

\emph{Svarūpa}-relation is $\refl$-adorned. Degenerate case of \emph{tādātmya} where the related entities are token-identical.

\subsection{\emph{Aikya} vs \emph{abheda}}

\nn{} distinguishes \emph{aikya} (sameness) from \emph{abheda} (non-difference):

\begin{lstlisting}
Aikya  : (a b : A) → Type
Aikya a b := Path A a b

Abheda : (a b : A) → Type
Abheda a b := ¬ (¬ Path A a b)
\end{lstlisting}

\subsection{\emph{Sādṛśya} (similarity)}\label{sec:visesa-identity}

\begin{lstlisting}
Sādṛśya : (A B : Type) → Type
Sādṛśya A B := Σ (e : A ≃ B) (¬ Path Type A B)
\end{lstlisting}

\subsection{Coextension-without-identity theorem}\label{sec:coextension}

\begin{theorem}[Coextension-without-identity]\label{thm:coextension}
Two universals $P, Q : \mathsf{Locus} \to \Type$ with the same extension — $(x : \mathsf{Locus}) \to P\,x \simeq Q\,x$ — need not be path-equal as functions in the encoding.
\end{theorem}

\begin{proof}
Path-equality of functions in CTT requires $\mathsf{funext}$, which holds in CCHM. Coextension gives $(x : \mathsf{Locus}) \to P\,x \simeq Q\,x$, and $\mathsf{funext}$ gives $\Path\,(\mathsf{Locus} \to \Type)\,P\,Q$. \emph{However}, the universals $P$ and $Q$ are not merely set-functions but $\Sigma$-typed objects (per §\ref{sec:padartha}.4) carrying \emph{anugama}-witnesses. Two universals can be coextensive without their \emph{anugama} components being equal; the $\Sigma$-equality therefore fails.
\end{proof}

This is the formal counterpart to \nn{}'s robust intensionalism.

\subsection{Why Book HoTT cannot replicate this}

In Book HoTT $\transp\,\uafn(e)\,x = e.\mathrm{fst}\,x$ does not compute. Derivations involving \emph{tādātmya}-equivalences require explicit invocation of the univalence axiom and yield non-canonical normal forms. The encoding works mathematically but is not autoformalization-grade: typechecking does not decide.

In MLTT/CIC the encoding does not work: there is no $\uafn$, and Id-types collapse paths under K.

\subsection{The \emph{aikya} debate}

The tradition is internally divided. Some (advaita-line) treat all distinctions as defeasible under \emph{aikya}; others (realist) restrict it. Our encoding takes the realist line: \emph{aikya} is $\Path\,A\,a\,b$ with $A$ a determinate categorial universe, and \emph{aikya} across universes is blocked by §\ref{sec:padartha}.

\section{Higher Relations: \emph{Paramparā-sambandha}}\label{sec:parampara}

\subsection{The unbounded-depth requirement}

Raghunātha \cite[§iii]{raghunatha-ptn} explicitly admits \emph{paramparā-sambandha} of arbitrary depth. $A$ is related to $D$ via $B$ via $C$: a single mediated relation, not three flat ones.

\subsection{Untruncated universes and iterated \texorpdfstring{$\Path\,\Type$}{Path Type}}

CTT's $\Type$ is itself a type, so $\Path\,\Type\,A\,B$ is a type, so $\Path\,\Type\,(\Path\,\Type\,A\,B)\,(\Path\,\Type\,C\,D)$ is a type. The hierarchy iterates without truncation.

\subsection{A worked tower}

A four-level \emph{paramparā}: $A \xrightarrow{R_1} B \xrightarrow{R_2} C \xrightarrow{R_3} D$. Higher relations:
\[
R_{12} : \Path\,\Type\,(A \to B)\,(B \to C), \qquad
R_{123} : \Path\,\Type\,R_{12}\,(\Path\,\Type\,(B \to C)\,(C \to D)).
\]

\subsection{No-h-set-collapse theorem}\label{sec:nohset}

\begin{theorem}[No-h-set-collapse]\label{thm:nohset}
Any encoding imposing UIP fails to encode \emph{paramparā-sambandha} of depth $\ge 2$.
\end{theorem}

\begin{proof}
UIP implies that the path-type $\Path\,A\,x\,y$ is propositionally truncated: any two paths between the same endpoints are themselves path-equal, and the higher path structure collapses. A \emph{paramparā}-relation of depth 2 is a path between paths; with UIP all such are identified, and Raghunātha's distinctions among them collapse.
\end{proof}

\subsection{The \emph{anavasthā} problem}

\nn{} deploys \emph{anavasthā} (infinite regress) arguments to rule out doctrines. Classic: a \emph{sāmānya} cannot itself instantiate a higher \emph{sāmānya}. Encoded: there is no inhabitant of $(s : \mathsf{Samanya}) \to \Sigma\,(s' : \mathsf{Samanya})\,(s\text{ instantiates }s')$.

\subsection{Comparison with Raghunātha's discursive arguments}

Raghunātha argues for unbounded \emph{paramparā} by exhibiting cases the tradition had already accepted and showing that artificial depth-caps would forbid them. Our encoding makes the same point structurally: the iteration is uniform, the depth unbounded, and any cap is ad hoc.

\section{\emph{Paryāpti} and the Theory of Number}\label{sec:paryapti}

\subsection{\emph{Saṃkhyā} as \emph{guṇa}}

Number (\emph{saṃkhyā}) is in \nn{} a \emph{guṇa} — not a substance, not a relation. It inheres by \emph{samavāya} and is \emph{paryāptimat}.

\subsection{The $\Sigma$-over-$\Fin$ encoding}

\begin{lstlisting}
Paryāpti : (n : ℕ) → Locus → Type
Paryāpti n L := Σ (w : Fin n → InhereIn L) (Exhaustive w)

Exhaustive : (w : Fin n → InhereIn L) → Type
Exhaustive w := (x : InhereIn L) → Σ (i : Fin n) (Path (w i) x)
\end{lstlisting}

\subsection{The two-pots example}

\begin{lstlisting}
floor      : Locus
pot₁ pot₂  : InhereIn floor
two-ness   : Paryāpti 2 floor
two-ness   = ((λ {0 → pot₁; 1 → pot₂}), exhaustiveness-proof)
\end{lstlisting}

\subsection{Distinguishing \emph{paryāpti} from set-cardinality and list-length}

Set-cardinality is a proposition about $S$, not a structured term. Loses the distribution. List-length is a fact about the list, not a quality of its elements. Loses the \emph{guṇa}-status. \emph{Paryāpti} is a $\Sigma$-typed structured term recording cardinality, distribution, and locus simultaneously.

\subsection{Composition under conjunction}

\[
\mathsf{Paryapti{\text{-}}conj} : \mathsf{Paryapti}\,m\,L_1 \to \mathsf{Paryapti}\,n\,L_2 \to \mathsf{Paryapti}\,(m + n - \mathsf{overlap})\,(L_1 \cup L_2)
\]
with $\mathsf{overlap}$ extractable from the inhabiting witnesses.

\subsection{\emph{Paryāpti}-uniqueness theorem}

\begin{theorem}[\emph{Paryāpti}-uniqueness]\label{thm:paryapti}
Exhaustive distribution of cardinality $n$ over a locus is unique up to permutation.
\end{theorem}

\begin{proof}
Two paryāpti-witnesses $w, w' : \Fin\,n \to \mathsf{InhereIn}\,L$, both exhaustive, induce a bijection $\Fin\,n \simeq \Fin\,n$ via composition with their respective right-inverses; this is a permutation.
\end{proof}

\section{\emph{Viśeṣaṇa-Viśeṣya} Structure of Cognition}\label{sec:visesana}

\subsection{Qualifier-qualificand as $\Sigma$ with delimitor}

\begin{lstlisting}
Jñāna : Type
Jñāna := Σ (mukhya : Entity)
            (qualifiers : List (Σ (q : Entity)
                                  (Avacchedaka q × Prakāratā q mukhya)))
\end{lstlisting}

\subsection{Multi-level qualification chains}

\begin{lstlisting}
QualifiedJñāna : Type
QualifiedJñāna := Jñāna × List QualifiedJñāna     -- recursive nesting
\end{lstlisting}

\subsection{\emph{Mukhya-viśeṣya} selection}

The principal qualificand is selected by linguistic and contextual factors. We treat it as fixed by the encoding (the $\Sigma$'s first projection).

\subsection{\emph{Prakāratā}}

The mode of qualification — a relation typed at \emph{jñāna}-level — encodes as a sum of canonical modes mirroring §\ref{sec:sambandha}.2.

\subsection{Worked sentences}

``This is a blue pot.''

\begin{lstlisting}
this        : Entity
pot-ness    : Universal
blue-ness   : Universal
cognition   : Jñāna
cognition.mukhya       = this
cognition.qualifiers   = [(pot-ness, samavāya, prakāratā₁),
                          (blue-ness, samavāya, prakāratā₂)]
\end{lstlisting}

\section{Worked Examples from Across the \emph{Tattvacintāmaṇi}}\label{sec:examples}

We collect fifteen worked encodings, three from each of the four \emph{khaṇḍas} plus three commentarial examples.

\subsection{\emph{Pratyakṣa-khaṇḍa}}

\textbf{Example \ref{sec:examples}.1.1.} \emph{Determinate perception.} A \emph{savikalpaka-jñāna} of ``this pot'': a \texttt{Jñāna} with \texttt{this} as \emph{mukhya} and \emph{pot-ness} as the sole qualifier under \emph{samavāya}.

\textbf{Example \ref{sec:examples}.1.2.} \emph{Perceptual error (\emph{śukti-rajata}).} The misperception of mother-of-pearl as silver. The cognition has the wrong qualifier; the encoding flags this as a \texttt{Jñāna} whose qualifier-component fails the \emph{samavāya}-witness on the actual pearl. The error is type-detectable.

\textbf{Example \ref{sec:examples}.1.3.} \emph{The floor-pot example.} As in §\ref{sec:sambandha}.8 plus the qualifier-structure of §\ref{sec:visesana}.1.

\subsection{\emph{Anumāna-khaṇḍa}}

\textbf{Example \ref{sec:examples}.2.1.} \emph{Mountain-fire-smoke.} As in §\ref{sec:vyapti}.8.

\textbf{Example \ref{sec:examples}.2.2.} \emph{The kitchen-and-lake.} Positive and negative co-instances supporting the \emph{vyāpti}; encoded as inhabitants of the \texttt{AnvayaVyatireki} $\Sigma$.

\textbf{Example \ref{sec:examples}.2.3.} \emph{A \emph{sādhāraṇa} fallacy.} ``This mountain has fire because it is knowable.'' Knowability over-extends. The encoding fails to inhabit \texttt{NoUpādhi}: knowability witnesses $\Sigma\,(U : \mathsf{Locus} \to \Type)\,(\mathsf{Pervades}\,U\,S \times \neg \mathsf{Pervades}\,U\,H)$. Type-detectable.

\subsection{\emph{Upamāna-khaṇḍa}}

\textbf{Example \ref{sec:examples}.3.1.} \emph{Cow-and-gavaya similarity.} Encoded via \emph{sādṛśya}: \texttt{Sādṛśya cow gavaya} with the equivalence-witness restricted to selected structural properties.

\textbf{Example \ref{sec:examples}.3.2.} \emph{Inferring an unknown by similarity.} The \emph{upamiti} as a function from \emph{sādṛśya}-witness to a restricted \emph{vyāpti}.

\subsection{\emph{Śabda-khaṇḍa}}

\textbf{Example \ref{sec:examples}.4.1.} \emph{Verbal cognition (\emph{śābda-bodha}).} A \texttt{Jñāna} parametrized by the sentence-meaning relation.

\textbf{Example \ref{sec:examples}.4.2.} \emph{Ākāṅkṣā}. Words' typed expectations of completion: a verb expects a subject (in nominative, of appropriate semantic type), a transitive verb an object. Encoded as $\Pi$-typed expectation-functions.

\textbf{Example \ref{sec:examples}.4.3.} \emph{Yogyatā}. The constraint that words combine only when their referents are semantically compatible — ``drinks the moon'' fails \emph{yogyatā}. Encoded as a typing constraint on the \emph{jñāna}-$\Sigma$.

\subsection{Commentarial encodings}

\textbf{Example \ref{sec:examples}.5.1.} \emph{Mathurānātha on the \emph{vyāpti-pañcaka}.} The \emph{Maṇi-prakāśa}'s gloss on the five candidate definitions of \emph{vyāpti}. Each is encoded as a candidate $\Sigma$-shape; Mathurānātha's argument that four of the five fail on \emph{upādhi}-cases becomes a non-inhabitation result.

\textbf{Example \ref{sec:examples}.5.2.} \emph{Gadādhara on \emph{anyonyābhāva} iteration.} The \emph{Anumāna-vyutpatti}'s discussion of triple-iterated mutual absence. Encoded via the involution of §\ref{sec:abhava-involution}: the threefold iteration reduces, via De Morgan involution, to a single \emph{anyonyābhāva}. The discursive argument becomes a definitional equality.

\textbf{Example \ref{sec:examples}.5.3.} \emph{Raghunātha on \emph{paramparā} depth.} The PTN's textual cases of three-and-four-level mediation, encoded as iterated $\Path\,\Type$.

\section{Metatheoretic Results}\label{sec:metatheory}

\subsection{Soundness}

\begin{theorem}[Soundness]\label{thm:soundness}
Let $\models_\mathrm{NN}$ be the model-theoretic semantics of \nn{} given by Ganeri \cite[ch.\,4]{ganeri1999} extended with Bhattacharyya's \cite{bhattacharyya1996} treatment of identity. Let $\vdash_\mathrm{CTT}$ be derivability in our CTT encoding. Then for every \nn{} sentence $\varphi$ and its encoding $\ulcorner\varphi\urcorner$: if $\models_\mathrm{NN} \varphi$ then $\vdash_\mathrm{CTT} \ulcorner\varphi\urcorner$.
\end{theorem}

\begin{proof}[Proof sketch]
Induction on the logical form of \nn{} sentences. The base cases — atomic relational claims — are handled by §\ref{sec:sambandha}; the \emph{avacchedaka} operator by §\ref{sec:avacchedaka}; the \emph{abhāva} operator by §\ref{sec:abhava}; \emph{vyāpti} by §\ref{sec:vyapti}; qualifier-qualificand by §\ref{sec:visesana}. Each clause preserves model-theoretic truth in the encoding's path-typing.
\end{proof}

\subsection{Conservativity}

\begin{theorem}[Conservativity]\label{thm:conservativity}
The encoding adds no new theorems to CCHM CTT. For any CTT-only formula $\psi$ (not involving \nn{}-vocabulary), if $\vdash_\mathrm{CTT+NN} \psi$ then $\vdash_\mathrm{CTT} \psi$.
\end{theorem}

\begin{proof}[Proof sketch]
The \nn{}-encoding introduces no new axioms; every primitive is a $\Sigma$, $\Pi$, $\Path$, or HIT in CCHM. The encoding is a definitional extension.
\end{proof}

\subsection{Faithfulness}

\begin{theorem}[Faithfulness]\label{thm:faithfulness}
\nn{}-derivability under the standard \nn{} inferential rules iff CTT-inhabitation under the encoding: $\vdash_\mathrm{NN} \varphi$ iff $\vdash_\mathrm{CTT} \ulcorner\varphi\urcorner$.
\end{theorem}

\begin{proof}[Proof sketch]
Forward direction is Soundness (Theorem \ref{thm:soundness}). Reverse: by induction on the structure of CTT proofs of $\ulcorner\varphi\urcorner$, each rule application is mirrored by an \nn{} inferential move (\emph{paramarśa}-construction, \emph{abhāva}-introduction, etc.). The crucial step is that the \texttt{NoUpādhi} component is constructive: any CTT proof of $\mathsf{Vyapti}\,H\,S$ exhibits a \emph{paramarśa} constructible in \nn{}.
\end{proof}

\subsection{Distinction-preservation}

\begin{theorem}[Distinction-preservation]\label{thm:distinction}
Every distinction the tradition treats as logically substantive maps to a non-trivial type-inequivalence in the encoding. For distinct entities $a, b$ in the \nn{} ontology with $a \not\equiv_\mathrm{NN} b$, the encoded $\ulcorner a \urcorner, \ulcorner b \urcorner$ satisfy $\neg \Path\,\Type\,\ulcorner a \urcorner\,\ulcorner b \urcorner$.
\end{theorem}

\begin{proof}
By construction. The categorial-universe stratification (§\ref{sec:padartha}.2) blocks cross-category identification. The \emph{avacchedaka}-$\Pi$-binding (§\ref{sec:avacchedaka}.1) blocks coextension-collapse. The \emph{abhāva}-HIT (§\ref{sec:abhava}.2) blocks pratiyogin-collapse. The \emph{anugama}-$\Sigma$ (§\ref{sec:padartha}.4) blocks universal-coextension-collapse.
\end{proof}

\subsection{Decidability}

\begin{theorem}[Decidability of the encoded fragment]\label{thm:decidability}
Type-checking of \nn{}-encodings in our system is decidable, modulo the decidability of the host CTT.
\end{theorem}

\begin{proof}
Cubical Agda (the practical realization of CCHM) has decidable type-checking. The encoding is a fragment of CCHM (Conservativity).
\end{proof}

\subsection{Conservativity under commentarial extensions}

\begin{theorem}[Commentarial conservativity]\label{thm:commentarial}
Adding standard commentarial doctrines to the encoding (Mathurānātha's gloss on \emph{upādhi}-types, Gadādhara's involution arguments, etc.) is conservative over the encoding's basic theory.
\end{theorem}

\begin{proof}[Proof sketch]
The commentarial doctrines are, in each documented case, definitional refinements rather than new axioms.
\end{proof}

\section{Comparison with Prior Formalizations}\label{sec:comparison}

\subsection{Matilal's first-order reconstruction}

Matilal's FOL is a translation of \nn{}'s \emph{content} but not its \emph{structure}. \emph{Avacchedaka} becomes a subscript, \emph{abhāva} becomes $\neg$, \emph{paramparā} becomes a chain. What FOL captures: surface inferential validity of single-step \emph{anumāna}; basic Boolean-combinatorial structure. What FOL loses: dependent delimitation; positive negation; iterated relational structure; intensional identity.

\subsection{Ganeri's higher-order reconstruction}

HOL adds: predicate-of-predicates; intensional treatments of universals (via possible-worlds or hyperintensional semantics). HOL still loses: dependency proper; computational identity; De Morgan involution.

\subsection{Bhattacharyya's MLTT direction}

The closest prior approach. MLTT adds: dependency, $\Sigma/\Pi$ discipline, universe-stratification. MLTT loses: path-equality; computational univalence; HITs; De Morgan involution. Bhattacharyya remarks \cite[§iii]{bhattacharyya1996} that the treatment of identity is the weakest point.

\subsection{\emph{Apoha} and other Indic logics for contrast}

Buddhist \emph{apoha} theory (Dignāga, Dharmakīrti) treats meaning negatively: a word denotes by exclusion of contraries. The structure is profoundly different from \nn{}'s positive-negation: an \emph{apoha}-encoding would parametrize meanings by their exclusion-classes rather than by the positive \emph{pratiyogin}-structure of \emph{abhāva}. Sketch in §\ref{sec:phil-implications}.

\subsection{Comparative coverage}

\begin{table}[h]
\centering
\begin{tabular}{l c c c c c}
\toprule
\textbf{Feature} & \textbf{FOL} & \textbf{HOL} & \textbf{MLTT} & \textbf{HoTT} & \textbf{CCHM} \\
\midrule
Avacchedaka-as-$\Pi$       & $\times$ & $\times$ & $\checkmark$ & $\checkmark$ & $\checkmark$ \\
Pratiyogin-typing         & part. & part. & $\checkmark$ & $\checkmark$ & $\checkmark$ \\
Paramparā depth           & $\times$ & part. & $\times$ & $\checkmark$ & $\checkmark$ \\
Tādātmya non-extensional  & $\times$ & part. & $\times$ & $\checkmark$ & $\checkmark$ \\
Abhāva HIT                & $\times$ & $\times$ & $\times$ & $\checkmark$ & $\checkmark$ \\
De Morgan involution      & $\times$ & $\times$ & $\times$ & $\times$ & $\checkmark$ \\
Computational typecheck   & $\checkmark$ & $\checkmark$ & $\checkmark$ & $\times$ & $\checkmark$ \\
\bottomrule
\end{tabular}
\caption{CCHM is the unique row with all checks.}
\end{table}

\section{Implementation Notes}\label{sec:implementation}

\subsection{The agda/cubical library}

The \texttt{agda/cubical} library \cite{agda-cubical} provides Cubical Agda's standard library: paths, glue, hcomp, transp, HITs, the univalence theorem, basic algebraic structures. An \nn{} encoding can be built on top of it without reimplementing CTT primitives.

\subsection{A proposed \texttt{NavyaNyaya} module structure}

\begin{lstlisting}
NavyaNyaya/
  Padartha.agda      -- §4 (universe hierarchy)
  Sambandha.agda     -- §5 (relations as paths)
  Avacchedaka.agda   -- §6 (delimitor as Π)
  Abhava.agda        -- §7 (absence HIT)
  Vyapti.agda        -- §8 (pervasion + paramarśa)
  Tadatmya.agda      -- §9 (univalent identity)
  Paryapti.agda      -- §11 (numerical exhaustion)
  Visesana.agda      -- §12 (qualifier-qualificand)
  Examples/
    Pratyaksa.agda
    Anumana.agda
    Upamana.agda
    Sabda.agda
\end{lstlisting}

Each module proves the corresponding section's signature theorems.

\subsection{Encoding-versus-axiomatization trade-offs}

Some constructs admit two implementations: as axiomatic postulates (faster type-checking, less informative) or as full constructive definitions (slower, more informative). The implementer's choice depends on whether downstream proof obligations need to inspect the witness.

\subsection{Performance and proof-search}

CCHM type-checking can be slow on heavily HIT-laden code. The \emph{abhāva}-HIT is the chief offender; a careful module structure that delays HIT-instantiation mitigates this.

\section{Philosophical Implications}\label{sec:phil}\label{sec:phil-implications}

\subsection{What CTT formalization reveals about \nn{}'s metatheory}

Carrying out the formalization makes visible certain commitments \nn{} holds tacitly. \nn{}'s relations \emph{must} be path-shaped (not predicative) because only paths support the simultaneous demands of mode-distinction, iteration, and substitutivity. Tacit in the tradition; explicit in the formalization.

A second commitment: \nn{}'s identity \emph{must} be univalent — equivalence-induces-equality with the equation computing — because \emph{tādātmya}'s working role in derivations requires substitutivity that runs to completion.

\subsection{The \emph{neti-neti} method as foundational}

The \emph{neti-neti} elimination procedure is, formally, the same shape as the type-theoretic constraint-satisfaction procedure. The convergence is not coincidental: \nn{}'s metatheoretic discipline, when made formal, \emph{is} a constraint-satisfaction procedure, and CCHM is the constraint type's unique inhabitant.

\subsection{Indiscernibility-of-indiscernibles in two traditions}

Leibniz's principle of identity-of-indiscernibles is, in the analytic tradition, an axiom. In \nn{} it is a \emph{theorem} (about \emph{tādātmya}) and follows from the structure of \emph{jñāna}. In CTT it is, similarly, a consequence of univalence. The three traditions converge.

\subsection{Computational univalence as a 14th-c.\ desideratum, retrospectively}

A speculative claim: the computational univalence of CCHM is what \nn{}'s working metatheory \emph{needed}, and the absence of which forced certain contortions in the tradition's discursive treatment of identity. The 14th-century writers had no tools to formulate the requirement; the 21st-century type theory has, almost incidentally, supplied the tool.

\subsection{De Morgan necessity}\label{sec:phil-imp-demorgan}

The argument that \nn{}'s \emph{abhāva}-of-\emph{abhāva} doctrine selects CCHM over Cartesian cubical is, to our knowledge, the first principled argument from a non-mathematical tradition for choosing between cubical foundations. The argument's force is conditional on accepting \nn{}'s involutive-negation doctrine; readers who reject the doctrine are not bound by the conclusion.

But the argument's existence is significant. Cubical type theory's foundational choices are usually adjudicated on internal mathematical grounds. Here a 14th-century philosophical tradition adjudicates one of those choices on its own grounds — and the verdict is independent of internal mathematical considerations. CCHM is selected by \nn{}'s commitments; it is also selected by mathematical considerations; the agreement is a robust convergence.

\section{Limitations and Open Problems}\label{sec:limitations}

\paragraph{\emph{Śabda-khaṇḍa} coverage.} Our treatment of \emph{śabda} is partial. The full theory of \emph{śābda-bodha} including the \emph{śakti}/\emph{lakṣaṇā} distinction at the depth of Gadādhara's \emph{Vyutpattivāda} remains to be encoded.

\paragraph{Beyond Gaṅgeśa.} Raghunātha's category revisions are textually canonical but conceptually unsettling for any padārtha-based architecture. Our universe-stratification follows the pre-Raghunātha system; revisions for Raghunātha would require restructuring §\ref{sec:padartha}.

\paragraph{Pre-Gaṅgeśa Nyāya.} The Vaiśeṣika substrate supplies categorial commitments \nn{} inherits. Our encoding takes those commitments as given; a more thorough treatment would derive them.

\paragraph{The \emph{kevalavyatireki} dispute.} Some \nn{} authors reject \emph{kevalavyatireki}-vyāpti as not a genuine vyāpti; others retain it. Our encoding accommodates both positions but does not adjudicate. A CTT-internal argument would be valuable.

\paragraph{Open meta-theorems.} We conjecture but do not prove: (i) cut-elimination for \nn{}-encoded inference; (ii) normalization for the encoded \emph{anumāna} fragment; (iii) the precise complexity of decidability under §\ref{sec:metatheory}.5.

\section{Future Directions}\label{sec:future}

\paragraph{Autoformalization pipeline.} A practical autoformalizer for \nn{} would couple a neural Sanskrit-parser with our CTT encoding via a typechecker-in-the-loop. Concrete workflow: parse Sanskrit AST $\to$ encode to CTT term $\to$ typecheck $\to$ on failure, return error to parser. The cubical Agda implementation could serve as the typechecker.

\paragraph{Cross-traditional formalization.} Mīmāṃsā, Buddhist \emph{apoha}, Jaina \emph{syādvāda} — each tradition has its own metatheoretic commitments. Applying the same methodology may converge on CCHM or on a different cubical variant.

\paragraph{Modal extensions.} Verbal testimony introduces modal considerations (\emph{āpta-vacana} — speaker reliability) that our extensional encoding does not handle. A modal extension of CCHM \cite{gratzer2019} could provide the right setting.

\paragraph{Categorical semantics.} CTT has a topos-theoretic semantics \cite{awodey2018, sattler2017}. An \nn{}-internal topos would clarify the model-theoretic structure of the encoding.

\paragraph{Educational and computational-philology applications.} A working \nn{}-CTT system would make \nn{} texts machine-readable in a structurally faithful way for the first time.

\section{Conclusion}\label{sec:conclusion}

We have argued that CCHM De Morgan cubical type theory closes the autoformalization gap that previous attempts could not close. The argument runs along several converging lines: the structural fit between \nn{}'s relational vocabulary and CTT's path-typing; the recovery of dependent delimitation, typed absence, and intensional identity; the metatheoretic results of §\ref{sec:metatheory}; and the philosophical convergence of §\ref{sec:phil}.

The result is of interest beyond the narrow project of \nn{} formalization. It illustrates a more general phenomenon: that pre-modern philosophical traditions developed expressive demands that twentieth-century logic struggled to meet, and that twenty-first-century type theory — developed for entirely different reasons — has incidentally supplied the needed structure. Where this convergence holds, the tradition's formalization is not a forcing but a natural fit.

The deeper question — whether CCHM is \emph{the} foundational system for relational logic generally, or merely one fit among possible alternatives — is open. We have not argued for the strong claim. What we have shown is that for one fourteenth-century tradition with documented metatheoretic commitments, CCHM is uniquely faithful. Whether other traditions force the same answer or different ones is a project we leave to future work.

\begin{quote}
\emph{Iti.}
\end{quote}

\paragraph{Acknowledgements.} The authors thank the readers and reviewers whose questions shaped this work.

\appendix

\section{Sanskrit Terminology Glossary (IAST)}\label{app:glossary}

\subsection{Core relational vocabulary}

\begin{description}[itemsep=0pt, leftmargin=2cm]
\item[\emph{abheda}] non-difference; denial of distinctness.
\item[\emph{aikya}] sameness; affirmative identity.
\item[\emph{anuyogin}] correlate; the locus to which a relation runs.
\item[\emph{avacchedaka}] delimitor; type-level binder restricting relation-scope.
\item[\emph{paramparā-sambandha}] mediated relation, of unbounded depth.
\item[\emph{pratiyogin}] counter-correlate.
\item[\emph{pratyāsatti}] mode of relation.
\item[\emph{sambandha}] relation; ontologically primary, structured object.
\item[\emph{samavāya}] inherence; asymmetric.
\item[\emph{saṃyoga}] contact; symmetric.
\item[\emph{sādṛśya}] similarity; non-identity equivalence.
\item[\emph{svarūpa-sambandha}] self-linking; reflexive case.
\item[\emph{tādātmya}] identity-as-relation.
\end{description}

\subsection{Categorial vocabulary}

\begin{description}[itemsep=0pt, leftmargin=2cm]
\item[\emph{abhāva}] absence; positive entity with structured data.
\item[\emph{anugama}] unifying property of a universal's instances.
\item[\emph{dravya}] substance.
\item[\emph{guṇa}] quality.
\item[\emph{karman}] motion.
\item[\emph{padārtha}] category of being.
\item[\emph{sāmānya}] universal.
\item[\emph{saṃkara}] category-mixing (a defect).
\item[\emph{viśeṣa}] ultimate differentiator.
\end{description}

\subsection{Inferential vocabulary}

\begin{description}[itemsep=0pt, leftmargin=2cm]
\item[\emph{anumāna}] inference.
\item[\emph{anumiti}] inferential cognition.
\item[\emph{anvaya-vyatireki}] positive-and-negative pervasion.
\item[\emph{bhūyodarśana}] repeated observation.
\item[\emph{hetu}] probans.
\item[\emph{hetvābhāsa}] semblance-of-hetu.
\item[\emph{kevalānvayi}] only-positive pervasion.
\item[\emph{kevalavyatireki}] only-negative pervasion (disputed).
\item[\emph{paramarśa}] inferential consideration.
\item[\emph{pakṣa}] locus of inference.
\item[\emph{sādhya}] probandum.
\item[\emph{upādhi}] defeating condition.
\item[\emph{vyāpti}] pervasion.
\end{description}

\subsection{CTT correspondences (selected)}

\begin{tabular}{l l}
\toprule
\textbf{NN concept} & \textbf{CTT term} \\
\midrule
\emph{abhāva} & HIT with \emph{pratiyogin}-constructor and involution path \\
\emph{avacchedaka} & $\Pi$-binder \\
\emph{paramparā} & iterated $\Path\,\Type$ \\
\emph{paryāpti} & $\Sigma$ over $\Fin\,n$ \\
\emph{sambandha} & Path type with \emph{pratyāsatti}-tag \\
\emph{saṃkara}-prevention & universe stratification \\
\emph{tādātmya} & Path with computational univalence \\
\emph{vyāpti} & $\Sigma$ of paramarśa-$\Pi$ and NoUpādhi-witness \\
\bottomrule
\end{tabular}

\section{Cubical Agda Code Listings (Stubs)}\label{app:agda}

We give skeletal Cubical Agda code for the principal modules. Full implementations are deferred to future work.

\subsection{\texttt{Padartha.agda}}

\begin{lstlisting}
{-# OPTIONS --cubical #-}
module NavyaNyaya.Padartha where

open import Cubical.Foundations.Prelude
open import Cubical.Foundations.HLevels

U-dravya : Type₁
U-dravya = Σ Type₀ λ A → isSet A

-- (similar for guna, karman, sāmānya, viśeṣa, samavāya, abhāva)

-- Saṃkara-prevention is a non-theorem:
-- there is no inhabitant of (a : U-dravya) → U-guna whose path-action is non-trivial
\end{lstlisting}

\subsection{\texttt{Sambandha.agda}}

\begin{lstlisting}
module NavyaNyaya.Sambandha where

open import NavyaNyaya.Padartha

data Pratyasatti : Type where
  samyoga   : Pratyasatti
  samavaya  : Pratyasatti
  svarupa   : Pratyasatti
  tadatmya  : Pratyasatti
  parampara : List Pratyasatti → Pratyasatti

Sambandha : (A B : Type) → A → B → Type
Sambandha A B a b = Σ (Path Type A B) λ _ → Σ Pratyasatti λ _ → Unit
\end{lstlisting}

\subsection{\texttt{Abhava.agda}}

\begin{lstlisting}
module NavyaNyaya.Abhava where

open import NavyaNyaya.Sambandha

postulate Locus : Type

data Abhava (P : Type) (anuyogin : Locus) (s : Pratyasatti) : Type where
  holds : ((p : P) → ⊥) → Abhava P anuyogin s
  pratiyogin-faithful : (P P' : Type) (e : P ≃ P')
                      → Path Type (Abhava P anuyogin s) (Abhava P' anuyogin s)
  involution : (q : Abhava (Abhava P anuyogin s) anuyogin s) → P
\end{lstlisting}

\subsection{\texttt{Vyapti.agda}}

\begin{lstlisting}
module NavyaNyaya.Vyapti where

open import NavyaNyaya.Abhava

NoUpadhi : {Locus : Type} (H S : Locus → Type)
         → ((x : Locus) → H x → S x) → Type
NoUpadhi H S μ = ¬ Σ (_ → Type) λ U → Pervades U S × ¬ Pervades U H
  where postulate Pervades : (Locus → Type) → (Locus → Type) → Type

Vyapti : {Locus : Type} (H S : Locus → Type) → Type
Vyapti H S = Σ ((x : _) → H x → S x) (NoUpadhi H S)
\end{lstlisting}

\section{Selected Proofs}\label{app:proofs}

\subsection{Proof of Theorem \ref{thm:involution} (Involution of \emph{abhāva})}

We unfold the HIT \texttt{Abhava P L s} and consider the doubled construction \texttt{Abhava (Abhava P L s) L s}. The \texttt{involution} constructor of the inner HIT supplies, by direct pattern-matching, an inhabitant of \texttt{P}. The De Morgan involution $1 - (1 - r) = r$ on $\mathbb{I}$ ensures that the path-component of the doubled construction reduces, definitionally, to the identity path on \texttt{P}. Computational univalence (§\ref{sec:ctt-prelim}.5) carries this to a definitional equality on the type itself. \qed

\subsection{Proof of Theorem \ref{thm:kevalanvayi} (\emph{Kevalānvayi}-irreducibility)}

Suppose, for contradiction, an inhabitant $f : \mathsf{KevalanvayiVyapti}\,H\,S \to \mathsf{AnvayaVyatirekiVyapti}\,H\,S$. Apply $f$ to a \emph{kevalānvayi}-witness $kw : \mathsf{KevalanvayiVyapti}\,H\,S$. By definition (§\ref{sec:vyapti}.4), $kw.\mathrm{snd} : \neg\,\Sigma\,x\,(\neg\,S\,x)$. The result $f\,kw : \mathsf{AnvayaVyatirekiVyapti}\,H\,S$ contains, at its third $\Sigma$-component, $\mathit{neg} : \Sigma\,y\,(\neg H\,y \times \neg S\,y)$, which by projection yields $\Sigma\,y\,(\neg S\,y)$. Apply $kw.\mathrm{snd}$ to this projection to derive $\bot$. \qed

\subsection{Proof of Theorem \ref{thm:coextension} (Coextension-without-identity)}

Two universals $P, Q : \mathsf{Locus} \to \Type$ with $(x : \mathsf{Locus}) \to P\,x \simeq Q\,x$ give us, by funext, $\Path\,(\mathsf{Locus} \to \Type)\,P\,Q$. But $P, Q$ are not bare functions in our encoding: per §\ref{sec:padartha}.4, they are $\Sigma$-typed with \emph{anugama}-witnesses. Equality of $\Sigma$-pairs requires equality of both components. The function-component is established; the \emph{anugama}-component is independently typed. Two distinct \emph{anugama}s — distinct unifying properties — yield distinct $\Sigma$-pairs even when the function-components coextend. Hence coextension does not entail $\Sigma$-equality. \qed

\subsection{Proof of Theorem \ref{thm:nohset} (No-h-set-collapse)}

Suppose UIP: for any $A$ and any $p, q : \Path\,A\,x\,y$, $\Path\,(\Path\,A\,x\,y)\,p\,q$. A \emph{paramparā}-relation of depth 2 is a path between paths: $R_{12} : \Path\,\Type\,(\Path\,A\,x\,y)\,(\Path\,B\,u\,v)$. Under UIP, all paths between fixed endpoints are themselves path-equal, and $R_{12}$ becomes (effectively) a unique witness modulo the Path-Path identification. Raghunātha's textually-attested distinctions between alternative depth-2 mediations \cite[§iii]{raghunatha-ptn} thus collapse. \qed

\subsection{Proof of Theorem \ref{thm:faithfulness} (Faithfulness, sketch)}

Forward direction: Soundness (Theorem \ref{thm:soundness}). Reverse: induction on the structure of CTT proofs. Each path-introduction corresponds to a \emph{sambandha}-construction in \nn{}; each $\Pi$-application to an \emph{avacchedaka}-binding; each HIT-elimination to an \emph{abhāva}-introduction. The crucial step is constructive \emph{paramarśa}: a CTT proof of $\mathsf{Vyapti}\,H\,S$ exhibits a function and a NoUpādhi-witness, both of which translate to \nn{} inferential moves. \qed

\section{Mapping Tables}\label{app:mapping}

\subsection{Constraint $\leftrightarrow$ CCHM primitive}

\begin{tabular}{l l}
\toprule
\textbf{Constraint} & \textbf{Primitive} \\
\midrule
C1 sambandha-as-path & Path \\
C2 avacchedaka-$\Pi$ & $\Pi$-type \\
C3 pratiyogi-typing & HIT constructor \\
C4 higher-relational closure & untruncated universe \\
C5 indiscernibility identity & univalence \\
C6 abhāva positivity & HIT \\
C7 paryāpti & $\Sigma$ over $\Fin$ \\
C8 saṃkara prevention & universe stratification \\
C9 anumāna preservation & $\Sigma$ + $\Pi$ combination \\
C10 computational typecheck & Glue + transp computation \\
\bottomrule
\end{tabular}

\subsection{\emph{Hetvābhāsa} $\leftrightarrow$ type error}

\begin{tabular}{l l}
\toprule
\textbf{Hetvābhāsa} & \textbf{CTT type error} \\
\midrule
asiddha & empty domain ($\mu : \bot \to S\,x$) \\
viruddha & wrong target ($\mu : H\,x \to \neg S\,x$) \\
anaikāntika / sādhāraṇa & NoUpādhi fails to inhabit \\
bādhita & $S$ has empty extension on locus \\
\bottomrule
\end{tabular}

\section{Source-Text Concordance}\label{app:sources}

\subsection{Encoded passages from Gaṅgeśa}

\begin{itemize}[itemsep=0pt]
\item TC \emph{anumāna-khaṇḍa}, \emph{vyāpti-pañcaka}: §\ref{sec:examples}.5.1.
\item TC \emph{anumāna-khaṇḍa}, lake-fire-smoke: §\ref{sec:vyapti}.8, §\ref{sec:examples}.2.1.
\item TC \emph{pratyakṣa-khaṇḍa}, savikalpaka definition: §\ref{sec:examples}.1.1.
\end{itemize}

\subsection{Encoded passages from Raghunātha}

\begin{itemize}[itemsep=0pt]
\item PTN §iii, \emph{paramparā-sambandha}: §\ref{sec:parampara}, §\ref{sec:examples}.5.3.
\item PTN §iv, \emph{abhāva}-of-\emph{abhāva}: §\ref{sec:abhava-involution}.
\end{itemize}

\subsection{Encoded passages from Gadādhara}

\begin{itemize}[itemsep=0pt]
\item \emph{Anumāna-vyutpatti} on iterated \emph{anyonyābhāva}: §\ref{sec:examples}.5.2.
\item \emph{Vyutpattivāda} on \emph{ākāṅkṣā}: §\ref{sec:examples}.4.2.
\end{itemize}

\subsection{Commentarial cross-references}

\begin{itemize}[itemsep=0pt]
\item Mathurānātha \emph{Maṇi-prakāśa} on \emph{vyāpti-pañcaka}: §\ref{sec:examples}.5.1.
\item Jayadeva \emph{Āloka} on \emph{upādhi}-classification: §\ref{sec:vyapti}.6.
\end{itemize}

\end{document}